\documentclass[preprint,pdflatex,sn-nature,oneside]{sn-jnl}

\usepackage[utf8]{inputenc}
\usepackage[english]{babel}
\usepackage[T1]{fontenc}

\usepackage{amsmath,amssymb}
\usepackage{mathtools}

\theoremstyle{thmstyleone}%
\newtheorem{theorem}{Theorem}
\theoremstyle{thmstyletwo}%
\newtheorem{remark}{Remark}%

\theoremstyle{thmstylethree}%
\newtheorem{definition}{Definition}%
\newtheorem{lemma}{Lemma}%

\makeatletter\@twosidefalse\@mparswitchfalse\makeatother

\usepackage{microtype} 

\usepackage{blindtext}

\usepackage{bm} 
\usepackage{graphicx}
\usepackage{xurl} 
\usepackage{stmaryrd}
\usepackage[table]{xcolor} 
\usepackage{caption}

\usepackage{booktabs,multirow,tabularx} 
\newcolumntype{C}{>{\centering\arraybackslash}X} 

\usepackage{enumitem} 
\usepackage{algorithm}
\usepackage{algorithmicx}
\usepackage{algpseudocode}

\usepackage{hyperref} 
\hypersetup{
  colorlinks=true,
  linkcolor=blue,
  filecolor=blue,
  citecolor=blue,      
  urlcolor=cyan,
}

\usepackage{todonotes}

\usepackage{etoolbox}

\usepackage{bibunits}
\defaultbibliography{cryptobib/abbrev3.bib,references,cryptobib/crypto.bib}
\defaultbibliographystyle{sn-nature}

\newtoggle{anonymous}
\newtoggle{comments}
\newtoggle{short}
\newtoggle{submission}
\togglefalse{anonymous}
\togglefalse{comments}
\togglefalse{short}
\togglefalse{submission}

\iftoggle{comments}{
  \newcommand{\daniel}[2][]{\todo[color=yellow!50,#1]{\textsf{Daniel:} #2}}
  \newcommand{\alex}[2][]{\todo[color=green!50,#1]{\textsf{Alex:} #2}}  
}{
  \newcommand{\daniel}[2][]{}
  \newcommand{\alex}[2][]{}
}

\iftoggle{submission}{
  \usepackage{lineno}
  \linenumbers}{
}

\iftoggle{anonymous}{
}{

\iftoggle{submission}{
\author[1]{\fnm{Alexander} \sur{Bienstock}*}
\author[1]{\fnm{Daniel} \sur{Escudero}*}
\author[1]{\fnm{Antigoni} \sur{Polychroniadou}*}

}{
\author[1]{\fnm{Alexander} \sur{Bienstock}}
\author[1]{\fnm{Daniel} \sur{Escudero}\footnote{Work done while at JPMorganChase. D.E.\@ is now at TACEO.}{\color{white},}}
\author[1]{\fnm{Antigoni} \sur{Polychroniadou}}
}

\author[1]{\fnm{Zhen} \sur{Zeng}}

\author[1]{\fnm{Pranav} \sur{Bhat}}

\author[1]{\fnm{Ashok} \sur{Singal}}

\author[1]{\fnm{Prashant} \sur{Sharma}}

\author[1]{\fnm{Manuela} \sur{Veloso}}

\affil[1]{\orgname{JPMorganChase}, \orgaddress{\city{New York}, \state{NY}, \country{USA}}}
}

\usepackage{stmaryrd}

\newcommand{\Func}[1]{\mathcal{F}_{\mathsf{#1}}}
\newcommand{\Prot}[1]{\Pi_{\mathsf{#1}}}

\mathchardef\mhyphen="2D 
\newcommand{\getsr}{\gets_{\$}}

\newcommand{\ID}{\mathsf{ID}}
\newcommand{\dist}{\mathsf{dist}}
\newcommand{\DB}{\mathsf{DB}}
\newcommand{\BioDist}{\mathcal{V}}

\newcommand{\LSH}{\mathsf{LSH}}

\newcommand{\Hash}{\mathsf{CH}}
\newcommand{\threshold}{\tau}

\newcommand{\SBA}{\mathsf{SBA}}
\newcommand{\Setup}{\mathsf{Bio\mhyphen Setup}}
\newcommand{\Enr}{\mathsf{Bio\mhyphen Enroll}}
\newcommand{\Auth}{\mathsf{Bio\mhyphen Auth}}
\newcommand{\pp}{\mathsf{pp}}
\newcommand{\pub}{\mathsf{pub}}

\newcommand{\HILL}{H_{\mathsf{HILL}}}
\newcommand{\aveminent}{\tilde{H}_{\infty}}

\newcommand{\Att}{\mathcal{A}}
\newcommand{\Sim}{\mathcal{S}}
\newcommand{\poly}{\mathsf{poly}}
\newcommand{\secparam}{\lambda}
\newcommand{\negl}{\mathsf{negl}(\secparam)}
\newcommand{\RO}{\mathsf{RO}}
\newcommand{\Exp}{\mathbb{E}}
\usepackage{tcolorbox}
\tcbuselibrary{skins,breakable}

\makeatletter
\newcommand{\DrawLine}{%
  \begin{tikzpicture}
  \path[use as bounding box] (0,0) -- (\linewidth,0);
  \draw[color=black!75,dashed,dash phase=2pt]
        (0-\kvtcb@leftlower-\kvtcb@boxsep,0)--
        (\linewidth+\kvtcb@rightlower+\kvtcb@boxsep,0);
  \end{tikzpicture}%
  }
\makeatother

\newtcolorbox{mybox}[2][]{%
  enhanced,
  title        = {#2},
  attach boxed title to top left={xshift=+3mm,yshift*=-3mm},
  breakable    = true,
  colback      = black!4,
  colframe     = black!75,
  fonttitle    = \bfseries,
  colbacktitle = black!10!white,
  coltitle     = black,
  #1
}

\newtcolorbox[auto counter,list inside=func]{functionality}[2][]{%
  enhanced,
  title        = {Functionality~\thetcbcounter: #2},
  attach boxed title to top left={xshift=+3mm,yshift*=-3mm},
  breakable    = true,
  colback      = yellow!4,
  colframe     = black!75,
  fonttitle    = \bfseries,
  fontupper    = \small,
  fontlower    = \small,
  colbacktitle = yellow!10!white,
  coltitle     = black,
  #1
}

\newtcolorbox[use counter=pro,list inside=prot]{protocol}[2][]{%
  enhanced,
  title        = {Protocol~\thetcbcounter: #2},
  attach boxed title to top left={xshift=+3mm,yshift*=-3mm},
  breakable    = true,
  colback      = black!4,
  colframe     = black!75,
  fonttitle    = \bfseries, 
  fontupper    = \small,
  fontlower    = \small,
  colbacktitle = black!10!white,
  coltitle     = black,
  #1
}

\newtcolorbox[use counter=pro,list inside=proc]{procedure}[2][]{%
  enhanced,
  title        = {Procedure~\thetcbcounter: #2},
  attach boxed title to top left={xshift=+3mm,yshift*=-3mm},
  breakable    = true,
  colback      = cyan!2,
  colframe     = black!75,
  fonttitle    = \bfseries, 
  fontupper    = \small,
  fontlower    = \small,
  colbacktitle = cyan!5!white,
  coltitle     = black,
  #1
}

\newtcbox{\xmybox}[1][red]{on line,
arc=7pt,colback=#1!10!white,colframe=#1!50!black,
before upper={\rule[-3pt]{0pt}{10pt}},boxrule=1pt,
boxsep=0pt,left=6pt,right=6pt,top=2pt,bottom=2pt}

\begin{document}
\setcitestyle{super,open={},close={}}

\title[Scalable Secure Biometric Authentication without Auxiliary Identifiers]{Scalable Secure Biometric Authentication without Auxiliary Identifiers}

\abstract{
  The prevalence of biometric authentication has been on the rise due to its ease of use and elimination of weak passwords.
  To date, most biometric authentication systems have been designed for on-device authentication of the device owner (e.g., smartphones and laptops).
  Recently, biometric authentication systems have started to emerge that are designed to authenticate users against cloud databases storing representations of biometrics for large numbers of users (potentially millions), such as those facilitating biometric payments.
  However, the use of a large cloud database introduces a significant attack vector, as a breach of the database could lead to the compromise of all enrolled users' sensitive biometric data.
  Indeed, all such existing systems either do not adequately protect against such a breach, or are impractical to deploy and use due to their high computational overhead.
  In this work, we present a new biometric authentication system that provides provable security guarantees against data breaches, while remaining scalable and performant.
  To do so, we marry artificial intelligence with advanced cryptographic techniques in a novel fashion, providing several optimizations along the way.
  Our work is the first to show that real-world scalable privacy-preserving biometric authentication without auxiliary identifiers is feasible, and we believe that it will spur widespread industrial adoption and further research in this area.
}
\maketitle

\iftoggle{submission}{\begin{bibunit}}{}
\section{Introduction}

Biometrics are becoming an increasingly popular method for user authentication, providing the means to unlock smartphones~\cite{abuhamad2020sensor}, access secure facilities~\cite{poh2009benchmarking}, board planes~\cite{khan2021use}, and more.
In such systems, a user is first enrolled by providing a biometric, and thereafter that same biometric (and sometimes only that biometric, without any auxiliary information) is used to authenticate the user.
Some examples of biometric modalities include face images, iris scans, fingerprints, palmprints, and more.

Biometric authentication will soon see expanded use-cases, including those where the biometric data of potentially hundreds of thousands (or millions) of users must be stored in a cloud database, rather than locally on some device.
Such use-cases include biometric-based payment systems~\cite{patra2022advancing}, which is being pushed as the payment method of the future by Amazon~\cite{amazonone} and JPMorgan~\cite{jpmorgan2023biometric}; fraud prevention via ensuring biometric uniqueness across accounts in a system (e.g., banking accounts); biometric access to cloud services~\cite{talreja2018biometrics}; and more.
However, the storage of large amounts of biometric data on a cloud server in these settings raises significant security and privacy concerns.
Indeed, even some of the world's most mature and trusted software companies are susceptible to data breaches~\cite{moore2017harms}, and biometric data is particularly sensitive since it is immutable and uniquely tied to an individual.
To give an illustrative example, if a person's credit card information is stolen, they can cancel the card and get a new one; if their biometric data is stolen, they cannot change something about their biology.

Once biometric data is compromised by an adversary, this adversary may try to recreate a physical copy of the biometric (for example, a fake fingerprint) in order to impersonate the victim.
However, this complex physical effort may not be needed; the adversary may be able to bypass the physical biometric scanner entirely and instead infiltrate the system past the point of physical scanning, using the stolen biometric at that point.
Moreover, since biometric data is soon to be used for financial transactions and other important services, the stakes are extremely high.
Beyond general privacy concerns, identity theft, fraud, surveillance, and other malicious activities are all potential risks if biometric data is not properly protected.
Furthermore, for companies offering biometric authentication services commercially, a data breach could lead to significant legal liability, reputational damage, and loss of customer trust~\cite{ponemon2024cost, ftc2023biometric}.
Therefore, ensuring the security of biometric authentication systems is of paramount importance.

A pipeline for processing biometric data typically involves converting the biometric into a so-called \emph{template}, which is a mathematical representation of the biometric that is suitable for storage and comparison.
While some organizations claim that simply storing templates in the cloud database is secure enough, there are numerous works demonstrating that templates can be inverted to recover the original biometric data~\cite{gomezbarrero2020inversebio, shahreza2023template, mai2018reconstruction} (see right side of Figure~\ref{fig:system}), and thus more work is needed.
Indeed, experts warn against the breach of biometric data from such systems, including that of Amazon~\cite{cnbc2023amazon}.

\begin{figure}[t!]
    \centering
    \includegraphics[width=\linewidth]{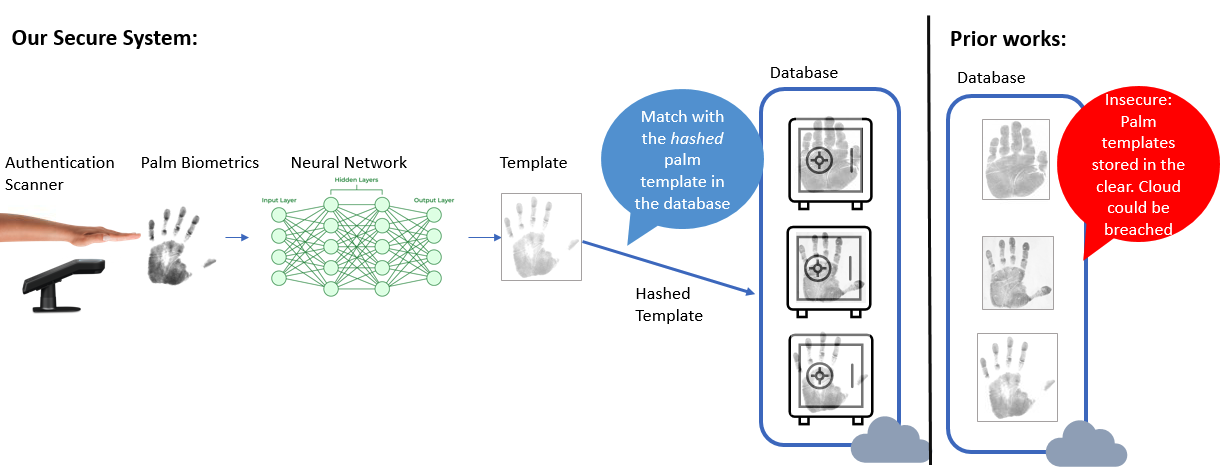}
    \caption{System architecture for our secure biometric authentication system. Templates computed from a neural network based on biometric scans are \emph{cryptographically hashed} before being stored in the database, which hides all information about the biometric data.
    These stored hashes are matched against for authentication.
    Previous solutions stored templates in the clear in the database, which can be inverted to recover original biometric data in the case of a data breach.}
    \label{fig:system}
\end{figure}

There are some prior works that provide \emph{provably secure} biometric authentication techniques.
First, some works use advanced cryptographic primitives like \emph{Fully-Homomorphic Encryption} and \emph{Secure Multiparty Computation}, which allow to perform the authentication process while the biometric data remains encrypted~\cite{boddeti2018secure, engelsma2022hers, choi2024blindmatch, CCS:BCHIKM25,EPRINT:BKSW24}.
However, these techniques all require some secret information to recover the result of the authentication, and this secret information must be stored somewhere.
The server storing the secret information can just as easily be compromised, leading to the same issues as before.
Additionally, these solutions are often impractical due to their high computational overhead~\cite{yang2023review}.

Another line of work uses \emph{fuzzy extractors}~\cite{dodFuzExt, JC:CFPRS21, CCS:SDFAMR25, ASIACCS:UYCKL21} to derive a strong cryptographic key from a biometric template, which is then used for authentication.
Although there is no need for a secret key to be stored on the server, fuzzy extractors are only used for \emph{one-to-one} verification of a user, in which the identity of the user is already known, explicitly through auxiliary user identifiers or implicitly through device ownership, and the purpose of the system is to verify the users's claimed identity.
A richer functionality that is required in many applications is \emph{one-to-many} identification, in which the identity of the user is not a priori known, and instead needs to be successfully determined from a large database of enrolled users, without any auxiliary identifiers.
Fuzzy extractors do not natively support this richer functionality.
Indeed, if a user \emph{only} presents a biometric, and not any other identifier, then a fuzzy extractor-based system must search through the entire large database of enrolled users one-by-one to find a match, which is computationally infeasible.

In our work, we present a new biometric authentication system that addresses the above challenges. 
See Figure~\ref{fig:system} for an overview of our system architecture.
Our system enables for the first time accurate, secure biometric authentication using a cloud database, without the need to store sensitive information on the server.
It supports real-time one-to-many identification with authentication time that scales independently of the number of enrolled users.
It additionally provides for the first time an important property called \emph{revocability}, which ensures that users can revoke their enrollment in the database in a way that guarantees their previously-enrolled biometric information cannot again be used by anyone, even if there was/is a prior/future database leak.
Our solution works for all biometric modalities, however we focus on face images in our experiments due to the availability of large-scale public datasets (and lack thereof for other modalities).
We leverage advanced cryptographic techniques to ensure the privacy and security of biometric data throughout the lifetime of the system, while also ensuring practicality for real-world use cases.
The properties achieved by our system in comparison to previous practical systems are listed in Table~\ref{tab:comparison}.

\begin{table}[h]
    \centering
    \begin{tabular}{l>{\centering\arraybackslash}p{0.22\linewidth}>{\centering\arraybackslash}p{0.22\linewidth}>{\centering\arraybackslash}p{0.18\linewidth}}
        \toprule
        & Encryption in Storage/Computation & Authentication Time & Revocability \\
        \midrule
        Existing Practical Systems & No & $O(N)$ & No \\
        Our Secure System & \textbf{Yes} & $\mathbf{O(1)}$ & \textbf{Yes} \\
        \bottomrule
    \end{tabular}
    \caption{Comparison of our secure biometric authentication system with existing practical systems. $N$ is the number of enrolled users.
    Our system encrypts biometric data at all times, has authentication time independent of number of enrollees, and provides revocability (defined above), while existing practical systems do not have any of these properties.}
    \label{tab:comparison}
\end{table}
\section{System Architecture and Threat Model}

We now overview the architecture of our system (depicted in Figure~\ref{fig:system}), along with the threat model.
A more formal model is presented in Appendix~\ref{sec:security}.
Our system consists of two types of devices: scanners that capture biometric data from users, and a remote database that receives information extracted from the biometric from the scanner, stores it and determines if authentication was successful.

During the enrollment phase, a user registers their biometric data in the system, in order to authenticate at some later point, using the same biometric.
The system stores some derivation of their biometric data associated with their identity in the database, alongside the (processed) data of all other users of the system.
Then, during the authentication phase, a user attempts to authenticate with the system using their biometric data.
The system checks whether the biometric data corresponds to a user in the database, and if so, returns the user's identity.

From a functionality standpoint, we naturally only want enrolled users to be able to successfully authenticate as themselves.
For security, we assume that the scanners are trusted devices that securely capture biometric data and send extracted information to the database over an encrypted channel.
In practice, the scanners immediately delete any data derived from the biometrics and would be fortified with anti-tampering measures~\cite{iso30107, marcel2023handbook}, liveness detection~\cite{marcel2023handbook}, and other security features to prevent adversarial manipulation.
On the other hand, the remote database may be compromised, and our goal is to ensure that even if an adversary gains access to the database, they cannot learn any sensitive information about enrolled users' biometrics.

\paragraph{Entropy Assumption.}
A primary assumption underlying secure biometric authentication systems is that biometric data is sufficiently random; i.e., has \emph{high entropy}.
Indeed, if biometric data was not very random, then an adversary could easily learn the distribution of biometric data, and simply sample likely biometric values until they find one that successfully authenticates as some enrolled user.
As explained in the Introduction, the adversary need not physically recreate the biometric, but could instead bypass the physical biometric scanner and infiltrate the system past that point, using the sampled biometric.
Fortunately, prior works have provided empirical evidence that biometric data does indeed have high entropy~\cite{CCS:SDFAMR25,Daugman}, and in Section~\ref{sec:exp} we provide further empirical validation of this assumption for our specific system instantiation.
In our system, we will in fact rely on this assumption of intrinsic randomness of biometrics to actually protect information about the biometric itself (see the next section and Appendix~\ref{sec:security}).\footnote{Indeed, we can tune the amount of entropy captured by our system, and thus the security of our system.}
\section{Protocol Overview}
\alex{More formal in appendix. No previous paper I saw had formal lists/pseudocode.}

\begin{figure}[t!]
    \centering
    \includegraphics[width=\linewidth]{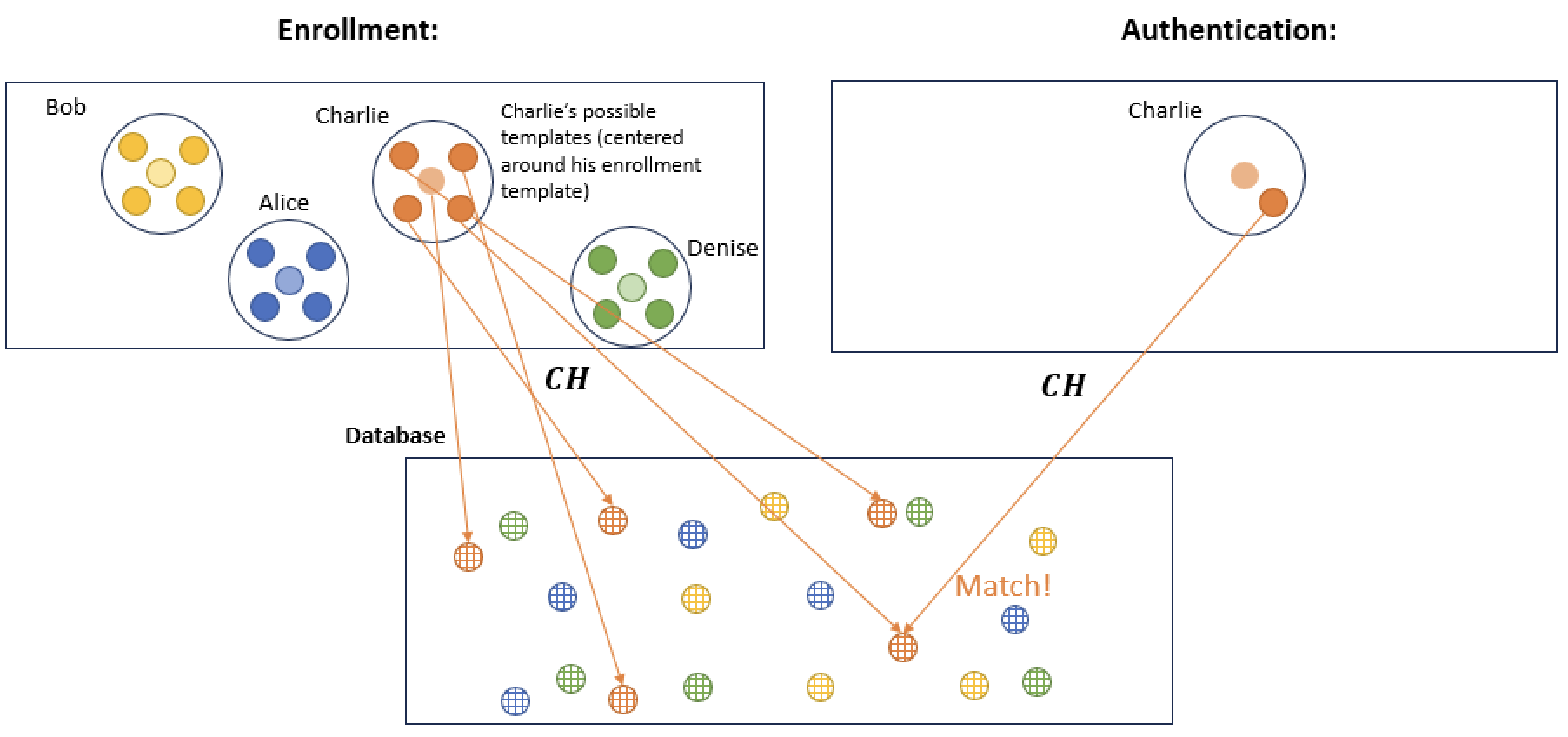}
    \caption{Pictorial, simplified overview of our secure biometric authentication protocol.
    During enrollment, the system cryptographically hashes all possible templates around the user's enrollment template (centered and faded) using $\Hash$ and stores these hash values in the database, paired with the user's identity (represented by color).
    During authentication, the system hashes the authentication template and checks if the resulting hash value matches any of the stored hash values for any user.
    Note that all hashed values stored in the database are randomly scattered throughout the space, depicting the security of the system.}
    \label{fig:alg_circles}
\end{figure}

Here, we provide a general overview of our system.
We start with the baseline one-to-one sample-then-lock fuzzy extractor of Canetti \emph{et al.}~\cite{JC:CFPRS21} (improved upon by Shukla \emph{et al.}~\cite{CCS:SDFAMR25}) and with novel techniques, extend it to the one-to-many authentication setting we aim to achieve in this work.
In Appendix~\ref{sec:formal_prot}, we provide a formal specification of our protocol and in Appendix~\ref{sec:security_proof} we provide a security proof for our protocol.

Our system critically relies on a base templatization model that extracts a high-dimensional, real-valued \emph{template} (or embedding) from a given biometric scan. This model ensures that templates extracted from biometrics of the same person are \emph{close}, while templates extracted from different people are \emph{far}, both in terms of \emph{euclidean distance} (which is closely related to \emph{cosine similarity}).

In addition, we require a \emph{cryptographic hash function}, $\Hash$, approved by the National Institutes of Standards and Technology, such as SHA-3~\cite{SHA3}, that maps bit strings to bit strings. The main property of a cryptographic hash function that we utilize is that if the input is sufficiently random (i.e., has high entropy), then the output will look random, and it will be computationally infeasible to recover any information about the input from the output. Cryptographic hash functions are \emph{deterministic}, meaning that the same input will always yield the same output.

With these two building blocks, we can provide the simplified intuition of our system, depicted pictorially in Figure~\ref{fig:alg_circles}.
The templatization model gives the property that all possible authentication templates for a given user will be closely centered around their enrollment template, and far away from the possible templates of all other users.
Thus, we can cryptographically hash all possible templates around each users' enrollment template, and store these hash values, paired with the user's identity, in the database.\footnote{Our final system performs this step in a more nuanced, efficient way.}
Then, during authentication, we can hash the authentication template and check if the resulting hash value matches any of the stored hash values for any user.
Correctness is achieved, since from above, this authentication hash value will match \emph{only} one of the hashes stored in the database that we computed during enrollment for that user, and not any other user.
Moreover, we achieve security assuming that biometric data has high entropy, since we only store outputs of the cryptographic hash function in the database.

Below, we introduce the system in more detail.
To bridge the gap between the real-valued templates and the discrete inputs of cryptographic hash functions, we also utilize a \emph{locality-sensitive hash} (LSH).
An LSH maps templates into bit strings, such that the output bit strings of two close templates are close in hamming distance (i.e., the number of bits where the strings differ), and two far templates result in bit strings that are far in hamming distance.

\paragraph{Global Setup.}
The system starts with some global setup.
We denote the output length of the LSH as $n$.
Additionally, let $0<k<n$ and $m>0$ be some other system parameters.
The system will first generate $m$ length-$k$ subsets of $[n]\coloneqq \{1,2,\dots,n\}$, denoted $\mu_1,\mu_2,\dots,\mu_m$.
An example for how to generate such subsets is just uniformly random and independently; see other methods in Section~\ref{sec:subsets}.
We will also use $t\coloneqq n-k$ to denote the number of elements of $[n]$ that we exclude from each subset $\mu_i$.

\begin{figure}[t!]
    \centering
    \includegraphics[width=\linewidth]{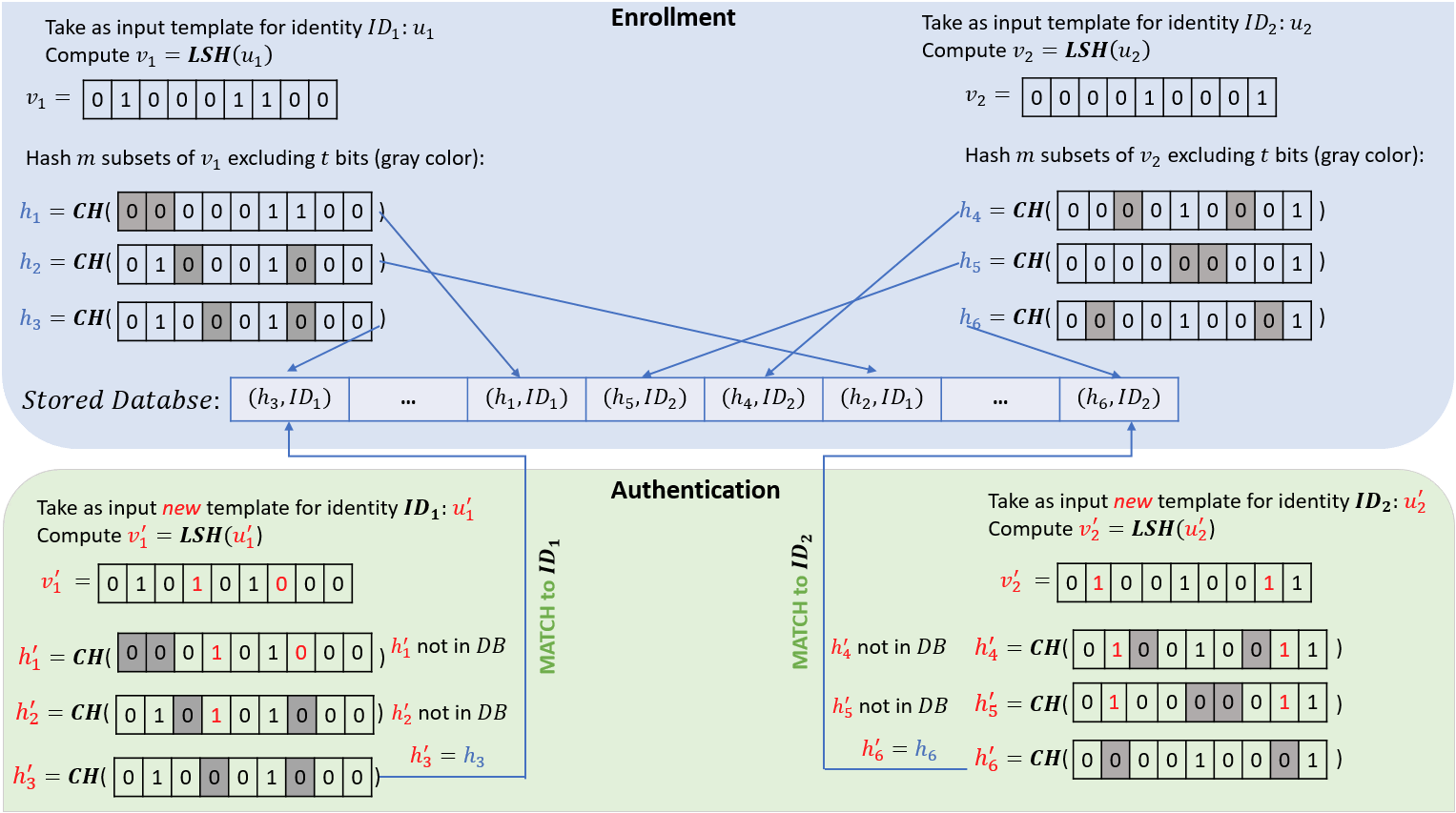}
    \caption{More detailed overview of our secure biometric authentication protocol, demonstrating the cryptographically hashed values of bit substrings that are stored in the database during enrollment (top) and against which authentication is performed (bottom).
    In this example, we have $m=3$, $n=9$, and $k=7$ (thus $t=2$).
    Since LSH output strings of the same user are close and LSH output strings of different users are far, the cryptographically hashed values of substrings will only match for same users, ensuring correctness.
    }
    \label{fig:alg_bits}
\end{figure}

\paragraph{Enrollment Phase.}
We now provide a more detailed description of the enrollment phase, demonstrated in the top of Figure~\ref{fig:alg_bits}.
After capturing the user's biometric data, the system runs it through the templatization model to obtain a template.
Next, this template is run through the LSH to obtain a \emph{bit string} of length-$n$, $v\coloneqq (v_1,\dots,v_n)$, where each $v_i\in\{0,1\}$.
With $v$ in hand, the system next computes all length-$k$ substrings of $v$ corresponding to $\mu_1,\mu_2,\dots,\mu_m$ generated during the global setup phase.
We denote these resulting substrings as $w_1,w_2,\dots,w_m$.
Finally, the system cryptographically hashes each $w_{i}$ using $\Hash$ to obtain a hash value $h_i\coloneqq \Hash(w_{i})$ and stores these hash values $h_1,\dots,h_m$ in a database, $\DB$, under the user's identity, $\ID$.
This database is a hash table, where the keys are the hash values and the values are the user identities; i.e. $(k,v)\coloneqq(h, \ID)$.
Such a hash table allows for constant-time lookups of keys (i.e., hash values).

\paragraph{Authentication Phase.}
Authentication follows a similar pattern to enrollment, presented in the bottom of Figure~\ref{fig:alg_bits}.
Indeed, authentication starts in the same way, by capturing the biometric, running it through the templatization model and LSH, obtaining substrings $w_1,w_2,\dots,w_m$ according to global subsets $\mu_1,\mu_2,\dots,\mu_m$, and hashing each substring to obtain hash values $h_1,h_2,\dots,h_m$.
Then, instead of storing the hashes as in enrollment, the system performs a lookup (in constant time) on each obtained hash $h_i$, within the database.
If the system finds a number of matching hashes for some identity $\ID$ that is greater than some threshold $\threshold$, it returns $\ID$; otherwise, it returns $\bot$ to indicate that authentication has failed.

\paragraph{Correctness.}
For the same person, the templatization model and LSH guarantee that different biometric scans result in bit strings $v$ and $v'$ that are close to each other, with high probability.
As a result, Canetti \emph{et al.}~\cite{JC:CFPRS21} observed that some number of the $m$ substrings $w$ of $v$ will match the corresponding substrings $w'$ of $v'$, with high probability.
Thus, with high probability, there will be enough hashed values during authentication time that match with the given user's enrolled hashes.
On the other hand, for different people, since the templatization model and LSH guarantee that their corresponding biometric scans will result in bit strings $v$ and $v'$ that are far apart, very few (if any) of the substrings of $v$ will exactly match any of the $m$ substrings of $v'$, with high probability.
Therefore, with high probability, not enough hashed value matches will be found during authentication time, and thus the system does not allow incorrect authentication.
See Section~\ref{sec:exp} for experimental verification of these phenomena.

\paragraph{Security.}
Security comes from our assumption that biometric data has \emph{high entropy}.
Indeed, if this is the case, and the parameters of the templatization model, LSH, and substring length $k$ are chosen appropriately, then each resulting bit string $w_i$ will also have high entropy.
Therefore, the cryptographic hash function will output \emph{random-looking} hash values $h_i$, that cannot be used to recover information about the original hashed substrings $v_{i}$ (nor the original biometric data or template).
In Section~\ref{sec:exp} we empirically validate these assumptions by employing standard entropy estimation~\cite{CCS:SDFAMR25,Daugman} to show that the hashed substrings indeed have high entropy.
\section{Concrete System Instantiation}\label{sec:concrete}
In this section, we give concrete instantiations of the abstracted components of our high-level protocol from the previous section.

\subsection{Templatization Model}
For our templatization model, we use Arcface~\cite{arcface} with 1024-dimensional templates.
Arcface adds additive angular margin loss to the loss function to maximize class separability, while pushing members of the same class closer together.
This increases accuracy on facial recognition tasks and intuitively also increases the amount of entropy that our system can capture.
Indeed, if templates for different identities are more separated and those for the same identity are closer together in euclidean space, this will translate to the corresponding LSH outputs having higher hamming distance and lower hamming distance for different and same identities, respectively.
Thus, longer substrings of these outputs will be much less likely to collide between different identities, while for the same identity, the likelihood of collision will remain high.
Therefore, we can sample longer substrings and thus capture more entropy in each of these substrings, while maintaining low error rates.

\subsection{Locality-Sensitive Hash Function}
For the LSH, we use the state-of-the-art OrthoHash~\cite{orthohash} with 4096-bit outputs.
Orthohash is based on \emph{deep hashing}---i.e., deep learning techniques, followed by some binarization process---instead of traditional statistical techniques like random projections.
The objective function used in OrthoHash simultaneously results in discriminative outputs, i.e., intra-class distances are small while inter-class distances are large, and also minimizes the quantization error from binarizing the continuous outputs.

\subsection{Subset Generation}\label{sec:subsets}
For the generation of subsets, we use \emph{$\zeta$-sampling} from Shukla \emph{et al.}~\cite{CCS:SDFAMR25}.
Instead of using the naive approach of sampling each bit uniformly at random for each substring, $\zeta$-sampling weights the bits according to some measure.
While, Shukla \emph{et al.} examine a few different such measures, we focus on a mutual-information based approach, which intuitively measures how much each bit in the LSH output reduces the uncertainty about the identity of the biometric:
$$\mathsf{mi}_i \coloneqq I(\mathsf{bit}_i;\mathsf{ID}) = H(\mathsf{ID}) - H(\mathsf{ID}|\mathsf{bit}_i),$$
where $H$ is the entropy function, $\mathsf{bit}_i$ is a random variable representing the value of bit $i$, and $\mathsf{ID}$ is a random variable representing the identity of the biometric.
Then, the sampling weight $\mathsf{wt}_i$ of each bit $\mathsf{bit}_i$ is proportional to this mutual information value, potentially further weighted by some value $\zeta$:
$$\mathsf{wt}_i \coloneqq \frac{\mathsf{mi}_i^\zeta}{\sum_j \mathsf{mi}_j^\zeta}.$$
This mutual-information based approach simultaneously increases accuracy (as bits that are more informative about the identity are more likely to be sampled) and also increases the entropy captured in each substring (as bits that are more discriminative are more likely to be sampled, and thus the length of substrings can be longer).
\section{Evaluation}
Now we provide empirical evaluation of our system, showing that it is efficient, accurate, and secure.

\subsection{Evaluation Dataset}
Due to the lack of available large-scale public datasets of high-quality face images with multiple images per identity, we use synthetic face images generated by the GAN-Control library~\cite{gancontrol}.
This library allows us to generate multiple realistic images of the same synthetic identity with control over facial expressions, lighting conditions, head poses, etc.
We generate two images each for approximately 100,000 synthetic identities.
See Figure~\ref{fig:faces} for examples of synthetic images used in our experiments.

\begin{figure}[t]
\centering

\textbf{Random Expressions}\\[0.5em]
\resizebox{\textwidth}{!}{%
\begin{tabular}{@{}c@{\hspace{0.1cm}}c@{\hspace{0.5cm}}c@{\hspace{0.1cm}}c@{\hspace{0.5cm}}c@{\hspace{0.1cm}}c@{\hspace{0.5cm}}c@{\hspace{0.1cm}}c@{\hspace{0.5cm}}c@{\hspace{0.1cm}}c@{}}
\includegraphics[width=0.09\textwidth]{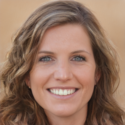} &
\includegraphics[width=0.09\textwidth]{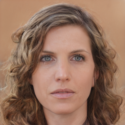} &
\includegraphics[width=0.09\textwidth]{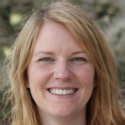} &
\includegraphics[width=0.09\textwidth]{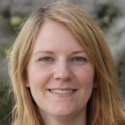} &
\includegraphics[width=0.09\textwidth]{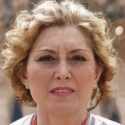} &
\includegraphics[width=0.09\textwidth]{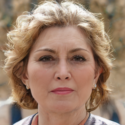} &
\includegraphics[width=0.09\textwidth]{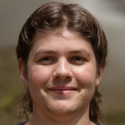} &
\includegraphics[width=0.09\textwidth]{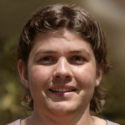} &
\includegraphics[width=0.09\textwidth]{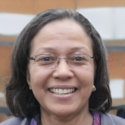} &
\includegraphics[width=0.09\textwidth]{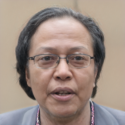} \\
\multicolumn{2}{c}{ID 1} &
\multicolumn{2}{c}{ID 2} &
\multicolumn{2}{c}{ID 3} &
\multicolumn{2}{c}{ID 4} &
\multicolumn{2}{c}{ID 5}
\end{tabular}}

\vspace{1em}

\textbf{Similar Expressions}\\[0.5em]
\resizebox{\textwidth}{!}{%
\begin{tabular}{@{}c@{\hspace{0.1cm}}c@{\hspace{0.5cm}}c@{\hspace{0.1cm}}c@{\hspace{0.5cm}}c@{\hspace{0.1cm}}c@{\hspace{0.5cm}}c@{\hspace{0.1cm}}c@{\hspace{0.5cm}}c@{\hspace{0.1cm}}c@{}}
\includegraphics[width=0.09\textwidth]{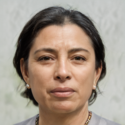} &
\includegraphics[width=0.09\textwidth]{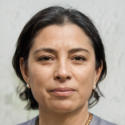} &
\includegraphics[width=0.09\textwidth]{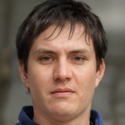} &
\includegraphics[width=0.09\textwidth]{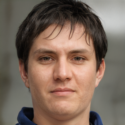} &
\includegraphics[width=0.09\textwidth]{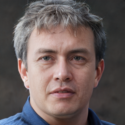} &
\includegraphics[width=0.09\textwidth]{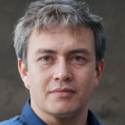} &
\includegraphics[width=0.09\textwidth]{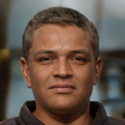} &
\includegraphics[width=0.09\textwidth]{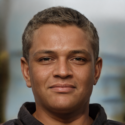} &
\includegraphics[width=0.09\textwidth]{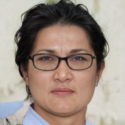} &
\includegraphics[width=0.09\textwidth]{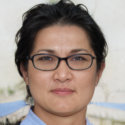} \\
\multicolumn{2}{c}{ID 1} &
\multicolumn{2}{c}{ID 2} &
\multicolumn{2}{c}{ID 3} &
\multicolumn{2}{c}{ID 4} &
\multicolumn{2}{c}{ID 5}
\end{tabular}}

\caption{Comparison of random expressions (top) and fixed expressions (bottom) for 5 different identities. Each ID shows two images side by side.}
\label{fig:faces}
\end{figure}

\subsection{Experimental Results}\label{sec:exp}
We now present the experimental results for our system, first explaining the insecure baseline we use for comparison, then describing the different types of experiments we run, and finally discussing the results of these experiments, including accuracy, computational efficiency, and entropy estimation.

\paragraph{Insecure Baseline.}
This insecure baseline simply stores the output of the ArcFace templatization model on the enrollment image of every identity, then during authentication, computes the euclidean distance between the template of the authentication image and the stored template for each enrolled identity, and authenticates if the distance is below some threshold.

\paragraph{Experiment Types.}
We run experiments on two sets of data:
(i) face images with different random expressions for each identity used during enrollment and authentication, and (ii) face images with two similar expressions for each identity used during enrollment and authentication, with the two expression styles fixed across identities.
See Figure~\ref{fig:faces} for examples of these two sets of data.
While (i) is a more challenging and realistic setting, publicly available templatization models are vastly inferior to (closed-source) industrial solutions~\cite{NISTFRTE}.
Moreover, as discussed earlier, the better the model is at separating the distance between same biometrics and different biometrics, the more entropy our system can capture, and thus the more secure it is.
Therefore, we use (ii) to come closer to matching the accuracy of industry models, and thus to show that our system can achieve very low error rates while also capturing a meaningful amount of entropy, which is crucial for security.
We still use (i) to empirically estimate the additional error that our secure system incurs compared to the insecure baseline of ArcFace, and to show that even in this more challenging setting, the error increase is not too large. 

In both settings, we evaluate the standard measures of False Negative Rate (FNR) and False Positive Rate (FPR).
FNR is the rate at which enrolled identities do not properly authenticate as themselves; FPR is the rate at which non-enrolled identities improperly authenticate as someone in the enrolled database.
We will also refer to generic \emph{accuracy}, for which we use the average of FNR and FPR as the \emph{error rate}.

\paragraph{Experiments with Random Expressions.}

We start by running experiments on synthetic faces generated with random expressions (top of Figure~\ref{fig:faces}).
As Extended Data Figure~\ref{fig:rand-error-rates} shows, the parameters of our system that yield the lowest error results in a $\lesssim 1.5\times$ increase in the error rate compared to the baseline insecure version;
this means that if the baseline insecure system has 99.9\% accuracy, our system will have $\approx 99.85\%$ accuracy.
Unfortunately, as mentioned earlier, the accuracy of even the baseline insecure system based on arcface for faces with random expressions is quite low.
Moreover, as a side effect, the difference between distances of same and different identities is not large, and thus the length of substrings and entropy captured at this accuracy is quite low (as explained above).
This low accuracy of the baseline system is due to a number of reasons, most notably due to the lack of large high-quality public datasets of face images needed for successful training;
it is known that industry models are able to achieve much better accuracy as high as $\approx 99.95\%$ for databases of size 12 million~\cite{NISTFRTE}.

\paragraph{Remedy for Low Captured Entropy.}
We propose a remedy to the above low entropy issue.
Indeed, although the length of the substrings and thus the entropy captured is low, if we sample a random cryptographic key and use it as input to a (also deterministic) Pseudo-Random Function (PRF)~\cite{GolGolMic86} to generate the outputs instead of a cryptographic hash function, then the outputs will be computationally indistinguishable from random, as long as the key remains secret.
One way to keep the key secret is to store it in a Trusted Execution Environment~\cite{munoz2023tee}, which is a secure environment isolated that is inaccessible from the rest of a server's operating system, and have all PRF computations happen inside this secure environment.
Thus, our security objective can still be achieved, albeit under a different assumption that the PRF key cannot be retrieved from the TEE.

\paragraph{Experiments with Similar Expressions.}
\begin{figure}[t]
    \centering
    \begin{minipage}{0.47\linewidth}
        \centering
        \includegraphics[width=\linewidth]{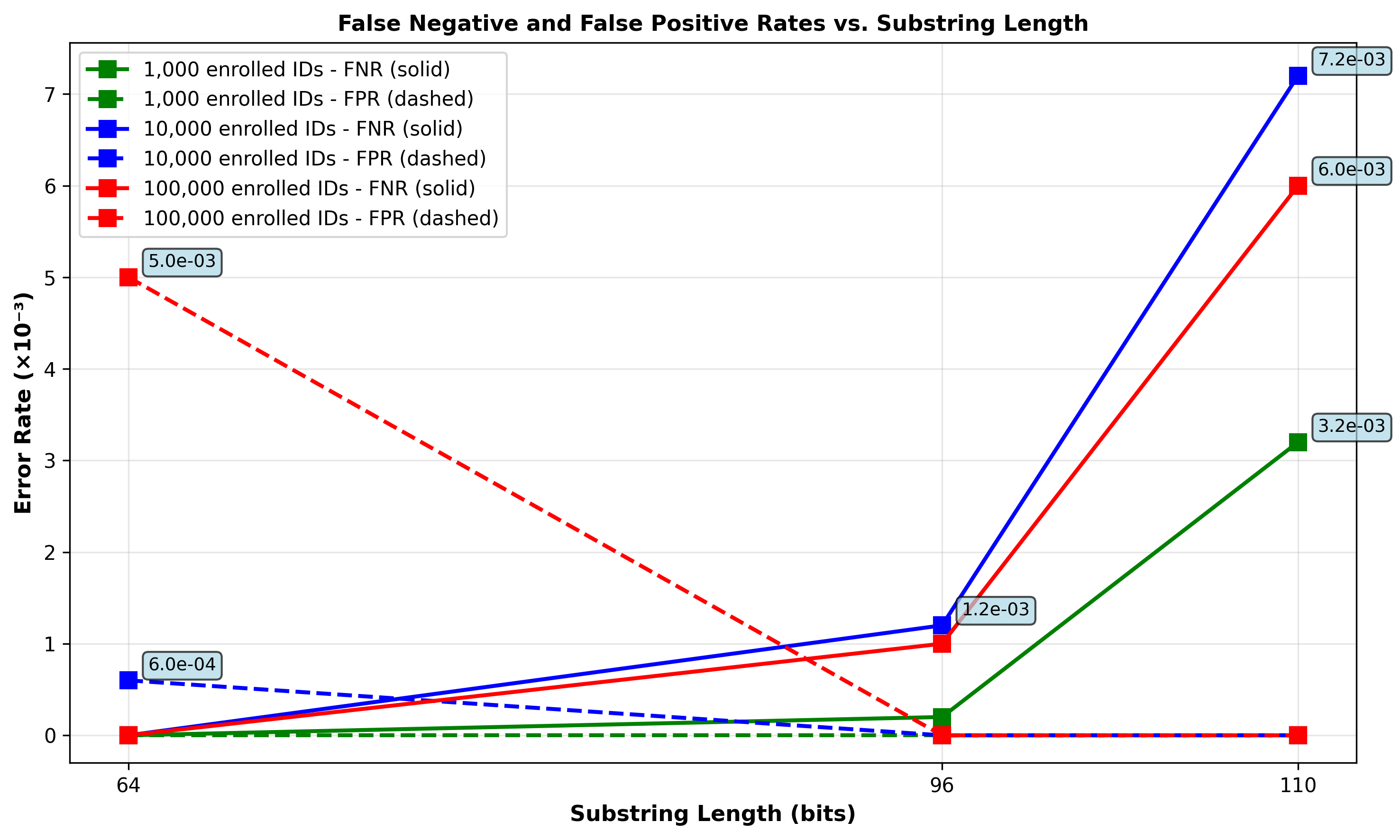}
    \end{minipage}
    \hfill
    \vrule
    \hfill
    \begin{minipage}{0.47\linewidth}
        \centering
        \includegraphics[width=\linewidth]{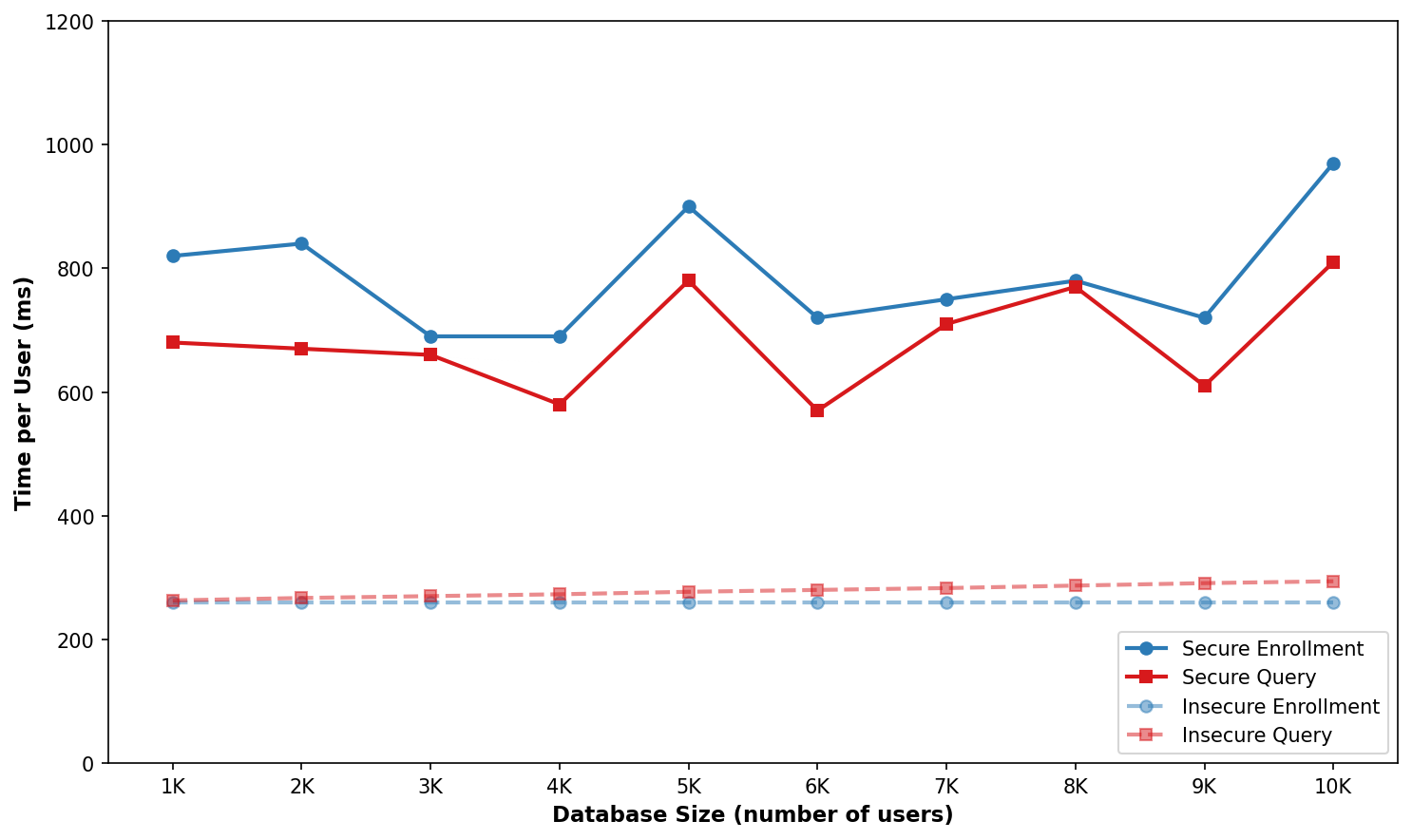}
    \end{minipage}
    \caption{\textbf{Left}. Error rates for databases of size 1,000, 10,000, and 100,000 enrolled IDs. FNR: False negative rate; FPR: False Positive Rate. False negatives are calculated over first 1000 enrolled identities. False positives are calculated over 1000 non-enrolled identities. Runs of our secure system are averaged over 5 trials.
    \textbf{Right}. Average time per user for enrollment and authentication queries; both secure (solid) and insecure (dashed, faded).}
    \label{fig:experiments}
\end{figure}
To come closer to matching the accuracy of industry models, we run experiments on synthetic faces with fixed, similar expressions (bottom of Figure~\ref{fig:faces}).
As shown in Figure~\ref{fig:experiments} (left), our secure system can achieve very low error rates in this setting, across database sizes.
Indeed, even for substring lengths up to 110 bits, we maintain very low and performant error rates.
As we increase the substring length, increasing the security of the system, we see a gradual increase in \emph{only} false negatives, but the system remains largely performant.

\paragraph{Computational Efficiency.}
We now discuss computational efficiency of our system.
First, for storage per identity, our sytem involves storing $m$ key-value pairs, where the key is a hash value and the value is the identity of the user.
The hash value length we use is the standard 256 bits and we use 32-bit integers to represent identities.
Thus, the total storage per identity is $m\cdot (256+32)$ bits.
We use $m = 250$K for the experiments in Figure~\ref{fig:experiments} (left), which results in a storage of $\approx 9$ MB per identity, which is quite reasonable for modern servers.
However, due to memory constraints, we recommend sharding the database so that each shard contains the (hash value, identity) pairs for at most 10K identities.
Authentication queries would be processed by each shard in parallel, whose results would be processed together by an authentication server.

With this in mind, in Figure~\ref{fig:experiments} (right) we present the average time per user for both enrollment and authentication queries, for database sizes up to 10K.
Both enrollment and queries take under 1 second for all database sizes and queries take as little as 680 milliseconds for database size of 1K.
We note that even though in theory, enrollment and query time should be fixed across database sizes, since we are only performing constant-time hash table operations, both actually slightly increase with larger databases.
This is because as the enrollment database grows, the hash table used for storing and looking up biometric templates becomes larger, leading to increased CPU cache misses during dictionary lookups, and therefore slower main memory fetches.
Compared to the insecure baseline, our system maintains competitive enrollment and query times.

\paragraph{Entropy Estimation.}
Measuring the true entropy of samples from a distribution requires an exponentially large number of samples~\cite{STOC:ValVal11,valent}.
However, there are established methods for estimating the min-entropy of biometrics in the literature.
To estimate the min-entropy of the substrings that we hash in our system, we use the standard method of Shukla \emph{et al.}~\cite{CCS:SDFAMR25}, which is a modification of Daugman's test~\cite{Daugman}.
This estimation involves computations on the mean, $\mu_{\textsf{Unlike}}$, and standard deviation, $\sigma_{\textsf{Unlike}}$, of pairwise Hamming distances between the substrings of different identities in the dataset:
$$e = \frac{-\mu_{\textsf{Unlike}}(1-\mu_{\textsf{Unlike}})\log\max\{\mu_{\textsf{Unlike}},1-\mu_{\textsf{Unlike}}\}}{\sigma^2_{\textsf{Unlike}}}.$$

In Extended Data Table~\ref{tab:entropy}, we show that the minimum (most important for security guarantees) of estimates of min-entropy for all substrings is quite high, reaching 74.3 bits of entropy for substring lengths of 110 bits, which is a very meaningful level of security.
We also provide the mean and maximum estimates across substrings for reference.

\iftoggle{submission}{
}{
\section*{Disclaimer}
This paper was prepared for informational purposes by the Artificial Intelligence Research group of JPMorgan Chase \& Co. and its affiliates (``JP Morgan'') and is not a product of the Research Department of JP Morgan. JP Morgan makes no representation and warranty whatsoever and disclaims all liability, for the completeness, accuracy or reliability of the information contained herein. This document is not intended as investment research or investment advice, or a recommendation, offer or solicitation for the purchase or sale of any security, financial instrument, financial product or service, or to be used in any way for evaluating the merits of participating in any transaction, and shall not constitute a solicitation under any jurisdiction or to any person, if such solicitation under such jurisdiction or to such person would be unlawful.

}

\iftoggle{submission}{\putbib\end{bibunit}}{\bibliography{cryptobib/abbrev3,references,cryptobib/crypto}}

\newpage
\renewcommand{\appendixname}{}
\begin{appendices}
\iftoggle{submission}{\section*{Methods}}{}
\iftoggle{submission}{\begin{bibunit}}{}
\section{Security}\label{sec:security}

\alex[inline]{I think for now, our security model should just entail the adversary obtaining a leak of the database, and that's it? (i.e., no ability to query the authenticator?)
Doing something stronger certainly wouldn't hurt, but maybe a bit out of scope?}
Here we present our formal security model for biometric authentication systems.
We start by providing some preliminary definitions that we will use in our security model and proof.

\subsection{Model Preliminaries}
We start with the standard notion of \emph{min-entropy}, which is a measure of the unpredictability of a random variable.
Intuitively, the best that an adversary can do to guess $A$ is to guess the most likely value of $A$, which has probability $\max_a \Pr[A=a]$.
\begin{definition}[Min-Entropy]
    The \emph{min-entropy} of a random variable $A$ is defined as: 
    $$H_\infty(A) = -\log_2\left(\max_a \Pr[A=a]\right).$$
\end{definition}

Min-Entropy is only defined for a random variable in isolation.
To measure the min-entropy in the presence of other correlated random variables, i.e., information stored in a database, \cite{dodFuzExt} introduced the following notion of \emph{average min-entropy}.
This notion intuitively defines how unpredictable a random variable remains, even given the correlated information, e.g., the database, and thus corresponds to security of the database.
\begin{definition}[Average Min-Entropy]
    The \emph{average min-entropy} of a random variable $A$ conditioned on another random variable $B$ is defined as:
    $$\aveminent(A\mid B) = -\log_2\left(\Exp_{b\gets B}\left[\max_a \Pr\left[A=a\mid B=b\right]\right]\right) = -\log_2\left(\Exp_{b\gets B}\left[2^{-H_\infty(A\mid B=b)}\right]\right).$$
\end{definition}
See~\cite{dodFuzExt} for more details, including why it is important to take the $\log$ outside of the expectation, for cryptographic applications.

Average min-entropy is a purely information-theoretic notion of security, meaning that security holds even against adversaries with unbounded computational resources.
However, in this paper we use cryptographic tools that are secure against (weaker, but more realistic) polynomial-time adversaries to enable our more performant system.
Therefore, we use the computational version of average min-entropy, conditional HILL entropy, introduced in~\cite[Definition 3]{EC:HsiLuRey07} (which in turn is based on notions from~\cite{HILL,BarakEnt}).

First, we provide the standard notion of computational distance between distributions, commonly used in the cryptographic literature.
Typically, we call some system \emph{secure} if the computational distance between any information that the adversary can infer from the real system and any information that the adversary can obtain from an ideal system that is independent of the secret is \emph{negligible} (i.e., super-polynomially small) in some given security parameter, $\secparam$.
\begin{definition}[Computational Distance]
    Let $X$ and $Y$ be two distributions.
    For a distinguisher $D$, we write the \emph{computational distance} between $X$ and $Y$ as
    $$\delta^D(X,Y) = |\Pr[D(X)=1]-\Pr[D(Y)=1]|.$$

    We say that the distance $\delta^D(X,Y)$ is \emph{negligible} in some security parameter $\secparam$, or $\delta^D(X,Y) = \negl$, if for every polynomial $p(\cdot)$, there exists $\secparam_0$ such that for every $\secparam\geq \secparam_0$, we have $\delta^D(X,Y) < 1/p(\secparam)$.
\end{definition}
\begin{definition}[Conditional HILL Entropy]
Let $(W, S)$ be a pair of random variables.
$W$ has \emph{HILL entropy at least $m$ conditioned on $S$}, denoted
$\HILL(W \mid S)\geq m$ 
if there exists a joint distribution $(X,S)$ such that $\aveminent(X\mid S)\geq m$ and $\delta^D((W,S),(X,S)) = \negl$ for any probabilistic polynomial-time distinguisher $D$.
\end{definition}
Conditional HILL entropy says that no adversary can distinguish the real joint distribution of the biometric data $W$ and the database $S$ from an ideal distribution in which a different distribution $X$ is swapped in for $W$ that has high average min-entropy given the database $S$, meaning that $X$ is unpredictable even given the database.
Therefore, security of the biometric authentication system follows.

As a stepping stone to prove the full security of our biometric authentication system, we use a notion called \emph{virtual black box obfuscation}~\cite[Definition 2]{EC:LynPraSah04}.
In essence, each hash $h$ we store in the database can be thought of as an obfuscation of a point function $P_v$ defined by $P_v(x) = 1$ if $x=v$ and $0$ otherwise, where $v$ is the substring of the LSH output that we cryptographically hash to obtain $h$.
\begin{definition}[Virtual Blackbox Obfuscation in the Random Oracle Model]
    An oracle algorithm $\mathcal{O}$ is a \emph{virtual blackbox obfuscator in the random oracle model} for a function family $\mathcal{F}$ if for every probabilistic polynomial-time adversary $\Att$, there exists a probabilistic polynomial-time simulator $\Sim$ such that for every $F\in \mathcal{F}$, we have
    $$|\Pr[\Att^{\RO}(\mathcal{O}^{\RO}(F))=1]-\Pr[\Sim^F(1^\secparam) = 1]|\leq\negl,$$
    where $\RO$ is the random oracle.
\end{definition}

In, \cite[Lemma 6]{EC:LynPraSah04} it is indeed shown that $h\gets\Hash(v)$ which we store in our database is a virtual blackbox obfuscation of the point function $P_v$ in the random oracle model.
At a high level, the simulator $\Sim$ with blackbox access to $P_v$ can simply sample a random value $h$ and run the adversary $\Att$ on input $h$.
Proofs in the random oracle model typically allow the simulator to answer the adversary's queries to the random oracle $\RO$.
Thus, whenever $\Att$ queries the random oracle $\RO$ on some input $x$, $\Sim$ can query $P_v$ on $x$; if $P_v(x)=1$, then $\Sim$ can answer $\Att$'s query with $h$, and otherwise $\Sim$ can answer with a random value.
\begin{lemma}[Obfuscation of Point Functions in the Random Oracle Model]\label{lem:obf}
    Let $P_v:\{0,1\}^k\to\{0,1\}$ be a function defined by $P_v(x) = 1$ if $x=v$ and $0$ otherwise.
    Define $\mathcal{P} = \{P_v : v\in\{0,1\}^k\}$.
    Then there exists a virtual blackbox obfuscator $\mathcal{O}$ for $\mathcal{P}$ in the random oracle model with security loss $0$, which obfuscates given $P_v$ by storing $r=\RO(v)$ and on input $x\in\{0,1\}^k$, checks if $\RO(x)=r$; if so outputs $1$, else $0$.
\end{lemma}
Note that recent literature on Fuzzy Extractors often make use of \emph{digital lockers}~\cite{EC:CanDak08}, which are related to obfuscations of point functions, but obfuscate the function $P_{v, s}$ defined by $P_{v, s}(x) = s$ if $x=v$ and $\bot$ otherwise.

\subsection{Our Formal Model}
Finally, we present our formal model.
We call a biometric authentication protocol $\SBA = (\Setup, \Enr, \Auth)$ a \emph{Secure Biometric Authentication protocol with Cloud Database} if it satisfies the following security definition.
Intuitively, the protocol is secure if for every enrolled identity, the biometric data of that identity still has high HILL entropy even given the public parameters and the entire database of enrolled users, meaning that the biometric data is computationally unpredictable even given the database.
\begin{definition}[Secure Biometric Authentication with Cloud Database]
    Consider any $N=\poly(\secparam)$, and any probabilistic polynomial-time adversary $\Att$.
    Let $W_1,\dots, W_N\getsr\BioDist$ be random variables representing the biometric samples of $N$ identities.
    Let $\pp\getsr\Setup(1^\secparam)$ be the public parameters of the system.
    For $i\in [N]$, execute $\pub_i\getsr\Enr(\pp,W_i)$.\alex{Might not be that the ultimate databnase just consists of each $\pub_i$ separately}
    Biometric Authentication protocol $\SBA = (\Setup, \Enr, \Auth)$ is \emph{secure} if for every $i\in[N]$:
    $$\HILL(W_i \mid (\pp, \pub_1,\dots,\pub_N)) \geq \secparam.$$
\end{definition}

\begin{remark}[Correctness]
    One could include a formal correctness notion for biometric authentication as well, however, since our system only provides empirical evidence of correctness, we omit it.
\end{remark}
\section{Formal Protocol Specification}\label{sec:formal_prot}

In Algorithms~\ref{alg:setup}-\ref{alg:auth} below, we formally present our entire biometric authentication protocol.

For notation, we use 
$n$ for the output length of the LSH;
$m$ for the number of subsamples;
$k$ for the length of the subsamples.
We use the $\LSH$ to denote the used locality-sensitive hash and let $\Hash$ to denote the cryptographic hash function, which is modeled as a random oracle~\cite{CCS:BelRog93} in our security proof below.

%
%

The setup algorithm presented in Algorithm~\ref{alg:setup} samples $m$ random subsets of $[n]$ of size $k$ according to some distribution $\mathcal{D}$.
These subsets will correspond to the substrings of LSH output that we cryptographically hash during enrollment and authentication, as shown below.
\begin{algorithm}
\caption{$\Setup$}\label{alg:setup}
\begin{algorithmic}[1]
\Statex \textbf{Parameters:} Distribution $\mathcal{D}$ over $[n]$ assigning weight $d_j$ to each bit $j \in [n]$
\For{$i \in [m]$}
	\State Sample $\mu_i = \{j_1, \dots, j_k\} \subset [n]$ by drawing $k$ elements without replacement from $\mathcal{D}$
\EndFor
\State \Return $\pp \gets (\mu_1, \dots, \mu_m)$
\end{algorithmic}
\end{algorithm}

The enrollment algorithm presented in Algorithm~\ref{alg:enroll} takes as input a user identity $\ID$ and biometric data $x$, and outputs a set of hash values that are stored in the server's database and map to $\ID$.
The algorithm first runs the data $x$ through a templatization model $\mathcal{M}$ to obtain a template $u$, then applies the LSH to obtain a bit string $v$, and finally hashes the subsamples of $v$ according to the subsets $\mu_i$ obtained during setup, to obtain hash values $h_i$ that are sent to the server for storage.
\begin{algorithm}
\caption{$\Enr$}\label{alg:enroll}
\begin{algorithmic}[1]
\Statex \textbf{Parameters:} Templatization model $\mathcal{M}$, locality-sensitive hash $\LSH$
\State $u \gets \mathcal{M}(x)$ \Comment{Extract template from image}
\State\label{step:subsample_hash:lsh}$v \gets \LSH(u)$ \Comment{Map template to bit string}
\For{$i \in [m]$}\label{step:subsample_hash:hash_subsamples} \Comment{Hash each subsample}
	\State $v_{\mu_i} \gets v_{j_1} \| \dots \| v_{j_k}$ where $\mu_i = \{j_1, \dots, j_k\}$
	\State $h_i \gets \Hash(v_{\mu_i})$
\EndFor
\State Send $(h_1, \dots, h_m)$ to the server
\State Server inserts $(h_i, \ID)$ into hash map $\DB$ for each $i \in [m]$
\end{algorithmic}
\end{algorithm}

Finally, the authentication algorithm presented in Algorithm~\ref{alg:auth} takes as input biometric data $x$ and outputs a user identity $\ID$ if the authentication is successful, or $\bot$ if it fails.
The algorithm first runs the data $x$ through the same steps as enrollment to obtain hash values $h_i$, and then sends these hash values to the server, who looks them up in the database and returns the identity $\ID$ that has at least $\threshold$ matches, where $\threshold$ is a threshold parameter that can be tuned to achieve the desired trade-off between false positives and false negatives.
\begin{algorithm}
\caption{$\Auth(\pp, x)$}\label{alg:auth}
\begin{algorithmic}[1]
\Statex \textbf{Parameters:} Templatization model $\mathcal{M}$, locality-sensitive hash $\LSH$, threshold $\threshold$
\State $u \gets \mathcal{M}(x)$ \Comment{Extract template from image}
\State $v \gets \LSH(u)$ \Comment{Map template to bit string}
\For{$i \in [m]$} \Comment{Hash each subsample}
	\State $v_{\mu_i} \gets v_{j_1} \| \dots \| v_{j_k}$ where $\mu_i = \{j_1, \dots, j_k\}$
	\State $h_i \gets \Hash(v_{\mu_i})$
\EndFor
\State Send $(h_1, \dots, h_m)$ to the server
\State Initialize dictionary $\mathsf{counts}[\cdot] \gets 0$
\For{$i \in [m]$}
	\If{$h_i \in \DB$}
		\State $\mathsf{counts}[\DB[h_i]] \gets \mathsf{counts}[\DB[h_i]] + 1$
	\EndIf
\EndFor
\For{$\ID \in \mathsf{counts}$}
	\If{$\mathsf{counts}[\ID] \geq \threshold$}
		\State \Return $\ID$
	\EndIf
\EndFor
\State \Return $\bot$
\end{algorithmic}
\end{algorithm}

\section{Security Proof}\label{sec:security_proof}
Finally, we prove that our biometric authentication system presented in the previous section is secure.
\begin{theorem}
    Consider, $\SBA$ presented in Algorithms~\ref{alg:setup}-\ref{alg:auth}.
    Assume that for every $i\in[\eta]$, the distribution $\BioDist_i$ of the substring of LSH output $V$ corresponding to subset $\mu_i$ sampled during $\Setup$ of Algorithm~\ref{alg:setup} has min-entropy at least $\secparam$, i.e., $H_\infty(\BioDist_i)\geq \secparam$ for every $i\in[\eta]$.
    Then $\SBA$ is a secure biometric authentication protocol with external database.
\end{theorem}

\begin{proof}
    Fix some $j\in[N]$.
    Let $V'_{i}$ be a random variable drawn from the distribution of substrings of LSH outputs corresponding to subset $\mu_i$ that is independent of $(\pp, \pub_1,\dots,\pub_N)$ (and in particular, $\pub_j$).
    Then, we have $\aveminent(V'_i \mid (\pp, \pub_1,\dots,\pub_N)) = H_\infty(V_i) \geq \secparam$.

    It remains to show that $\delta((V_i, (\pp, \pub_1,\dots,\pub_N)), (V'_i, (\pp, \pub_1,\dots,\pub_N))) =\negl$.
    First, by the security of the implicit obfuscation of point functions in our construction (Lemma~\ref{lem:obf}), we have that for every distinguisher $D$,\alex{Formally, we would need to build a simulator $\Sim'$ that uses the obfuscation simulator $\Sim$}
    \begin{align}\label{eq:hyb1}
    |\Pr[D(V_i, & (\pp, \pub_1,\dots,\pub_N)) = 1] \notag \\
    & - \Pr[\Sim^{P_{V_i}}(\pp, \pub_1,\dots,\pub_{j-1}, \pub_j', \pub_{j+1},\dots,\pub_N) = 1]| = 0,
    \end{align}
    where $\pub_j'$ is the same as $\pub_j$, except it excludes the hash value corresponding to $v_i$.

    Next, since $H_\infty(\BioDist_i) \geq \secparam$, we have that
    \begin{align}\label{eq:hyb2}
    \Pr[\Sim^{P_{V_i}}(\pp, \pub_1,&\dots,\pub_{j-1}, \pub_j', \pub_{j+1},\dots,\pub_N) = 1] \notag \\
    &  - \Pr[D((V'_i, (\pp, \pub_1,\dots,\pub_N))) = 1] = \negl.
    \end{align}
    Indeed, $\Sim$ will never be able to guess $V_i$ and thus it as if all other information is independent of $V_i$, which is the case in the second term of Equation~\ref{eq:hyb2}.
    \alex[inline]{Could formally use FEAP loss here}

    Adding up equations~\ref{eq:hyb1} and~\ref{eq:hyb2}, we get negligible total computational distance.
\end{proof}






\iftoggle{submission}{\section*{Data Availability}
The datasets used for this study are available from the corresponding author on request.

\section*{Code Availability}
The code used for this study is available from the corresponding author on request.}{}

\iftoggle{submission}{\putbib\end{bibunit}}{
  \section{Extended Data}
}
\end{appendices}

\iftoggle{submission}{\section*{Contributions}
\alex[inline]{Check}
A.B.\@ wrote most of the paper, performed the experiments, and designed the system.
D.E.\@ assisted with the literature search.
A.P.\@ helped write the paper and managed the project.
A.S.\@ introduced the project to the team and provided feedback.

\section*{Competing Interests}
A.B., D.E., A.P. and A.S.\@ are co-inventors on a patent application related to this work, filed by JPMorganChase.

\section*{Correspondence and Requests for Materials}
Should be sent to Alexander Bienstock or Antigoni Polychroniadou.}{}

\iftoggle{submission}{\newpage}{}
\renewcommand{\figurename}{Extended Data Figure}
\setcounter{figure}{0}
\renewcommand{\thefigure}{\arabic{figure}}
\begin{figure}[t]
    \centering
    \includegraphics[width=0.85\linewidth]{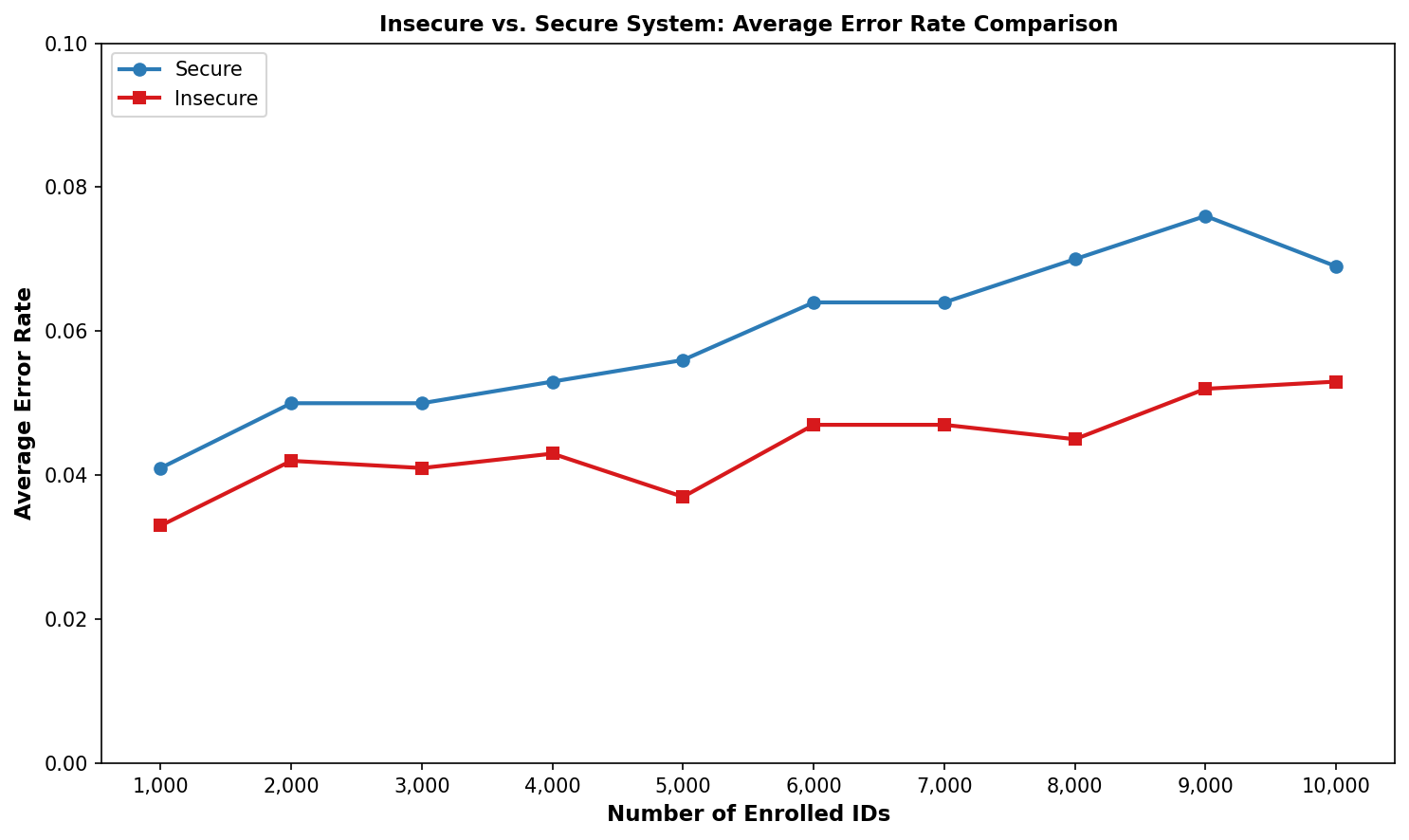}
    \caption{Error rates for databases of size 1,000-10,000 for both the secure and insecure baseline systems. False negatives are calculated over first 1000 enrolled identities. False positives are calculated over 1000 non-enrolled identities. We average false negative and false positive rates for clarity. Runs of our secure system are averaged over 5 trials.}
    \label{fig:rand-error-rates}
\end{figure}

\renewcommand{\tablename}{Extended Data Table}
\setcounter{table}{0}
\renewcommand{\thetable}{\arabic{table}}
\begin{table}[h]
\centering
\begin{tabular}{cccc}
\hline
\textbf{Substring Length (bits)} & \textbf{Min} & \textbf{Max} & \textbf{Mean} \\
\hline
64  & 47.28 & 56.71 & 52.35 \\
96  & 65.54 & 79.30 & 72.87 \\
110 & 74.3 & 87.70 & 80.93 \\
\hline
\end{tabular}
\caption{Statistics on min-entropy estimates across substring lengths.}
\label{tab:entropy}
\end{table}


\end{document}